\RequirePackage{snapshot} 
\documentclass[journal]{IEEEtran}%
\usepackage{amssymb}
\usepackage{amsfonts}
\usepackage{graphicx}
\usepackage{amsmath}%
\newtheorem{thm}{Theorem}
\newenvironment{theorem}{\bf\begin{thm}\rm\em}{\end{thm}} 
\newtheorem{cor}[thm]{Corollary}
\newenvironment{corollary}{\bf\begin{cor}\rm\em}{\end{cor}} 
\newtheorem{lem}[thm]{Lemma}
\newtheorem{prop}[thm]{Proposition}
\newenvironment{proposition}{\bf\begin{prop}\rm\em}{\end{prop}} 
\newtheorem{rem}[thm]{Remark}
\newenvironment{remark}{\bf\begin{rem}\rm}{\end{rem}} 

\ifCLASSINFOpdf
\usepackage[pdftex,colorlinks,
           bookmarks=true,%
           pdftitle=How\ user\ throughput\ depends\ on\ the\ traffic\ demand\ in\ large\ cellular\ networks%
           pdfauthor=B.\ Blaszczyszyn%
           \ -\ M.Jovanovi\ -\ c M.K.\ Karray,]{hyperref} 
\hypersetup{
   bookmarksnumbered,
   pdfstartview={FitH},
   citecolor={blue},
   linkcolor={red},
   urlcolor=[rgb]{0,0.55,0},
   pdfpagemode={UseOutlines}
}
\usepackage[numbers,sort&compress]{natbib}
\usepackage{hypernat}
\else

 \usepackage[dvips]{graphicx}
\usepackage[numbers,sort&compress]{natbib}
\fi

\renewcommand{\Pr}{\mathbf{P}}

\graphicspath{{./}{MeasurementsCellQoS/TypicalVsMean/}}

\newcommand\EXCLUDE[1]{}

\begin{document}

\title{How user throughput depends on the traffic demand in large cellular networks}
\author{\IEEEauthorblockN{Bart{\l }omiej~B{\l}aszczyszyn\IEEEauthorrefmark{1},
Miodrag Jovanovic\IEEEauthorrefmark{2}\IEEEauthorrefmark{1} and Mohamed Kadhem Karray\IEEEauthorrefmark{2} 
}
}

\maketitle
\let\thefootnote\relax\footnote{\IEEEauthorrefmark{1}INRIA-ENS,
23 Avenue d'Italie, 75214  Paris, France
Email: Bartek.Blaszczyszyn@ens.fr\\
\indent\IEEEauthorrefmark{2}Orange Labs;
38/40 rue G\'{e}n\'{e}ral Leclerc, 92794  
Issy-les-Moulineaux, France
Email: \{miodrag.jovanovic,\,mohamed.karray\}@orange.com}

\newcommand{\thefootnote}{\arabic{footnote}}
\addtocounter{footnote}{-1}

\thispagestyle{empty}

\begin{abstract}
We assume a space-time Poisson process of call arrivals on the infinite plane,
independently marked by data volumes and served by a cellular network modeled
by an infinite ergodic point process of base stations. Each point of this
point process represents the location of a base station that applies a
processor sharing policy to serve users arriving in its vicinity, modeled by
the Voronoi cell, possibly perturbed by some random signal propagation
effects. User service rates depend on their signal-to-interference-and-noise
ratios with respect to the serving station.

Little's law allows to express the mean user throughput in any region of this
network model as the ratio of the mean traffic demand to the steady-state mean
number of users in this region.
Using ergodic arguments and the Palm theoretic formalism, we define a global
mean user throughput in the cellular network and prove that it is equal to the
ratio of mean traffic demand to the mean number of users in the steady state
of the \textquotedblleft typical cell\textquotedblright\ of the network. Here,
both means account for double averaging: over time and network geometry, and
can be related to the per-surface traffic demand, base-station density and the
spatial distribution of the signal-to-interference-and-noise ratio. This
latter accounts for network irregularities, shadowing and cell dependence via
some cell-load equations.

Inspired by the analysis of the typical cell, we propose also a simpler,
approximate, but fully analytic approach, called the mean cell approach. The
key quantity explicitly calculated in this approach is the cell load. In
analogy to the load factor of the (classical) M/G/1 processor sharing queue,
it characterizes the stability condition, mean number of users and the mean
user throughput. We validate our approach comparing analytical and simulation
results for Poisson network model to real-network measurements.
\end{abstract}



\section{Introduction}

Mean user throughput is a key quality-of-service metric in cellular data
networks. It describes the average \textquotedblleft speed\textquotedblright%
\ of data transfer during a typical data connection. It is usually defined as
the ratio of the average number of bits sent (or received) per data request to
the average duration of the data transfer. Since coexisting connections in a
given network cell share some given cell transmission capacity, mean user
throughput depends inherently on the requested data traffic. It also depends
on the network architecture (positioning of the base stations) and in fact may
significantly vary across different network cells. Moreover, extra-cell
interference makes performance of different cells interdependent. Predicting
the mean user throughput in function of the mean traffic demand locally (for
each cell) and globally in the network (which involves appropriate spatial
averaging in conjunction with the temporal one, already present in the
classical definition of the throughput) is a key engineering task in cellular
communications. In this paper we propose an analytic approach to the
evaluation of the mean user throughput in large irregular cellular networks,
validated by real-network measurements performed in operational networks.


Little's law allows to calculate the mean user throughput as the ratio of the
mean traffic demand (number of bits requested per unit of time) to the mean
number of users in the steady state of the network. This argument can be used
to express mean user throughput locally in any region of the network.
Statistics usually collected in operational networks allow to estimate the
mean traffic demand and the mean number of users for each cell and hour of the
day. Even if they carry important information about the local network
performance, they exhibit important variability over time (24 hours) and
network cells; cf. Figure~\ref{f.crude-data}. This can be explained by the
fact that mean user throughput in a particular cell does not depend only on
the traffic in this cell, but also on the neighbouring cells. Moreover, the
geometry of different cells in a real network may significantly differ. For
these two reasons, the family of local (established for each cell)
throughput-versus-traffic laws usually exhibits a lot of variability both in
real data and in network simulations, and hence does not explain well the
macroscopic (network- level) relation between the mean traffic demand and mean
user throughput. Finding such a macroscopic relation is an important task for
network dimensioning.
It is clear that an appropriate spatial averaging
is necessary to discover such a macroscopic law.

Spatial averages of point patterns, modeling in our case the geographic
locations of base stations, can be studied using the formalism of Palm
distributions naturally related to the ergodic results for point processes.
Within this setting one considers a typical base station with its typical cell
(zone of service) whose probabilistic characteristics correspond to the
aforementioned spatial averages of the characteristics for all base stations
in the network. Adopting this formalism, we define the \emph{mean user
throughput in the infinite ergodic network} as the limit of the ratio of the
mean volume of the data request to the mean service duration in a large,
increasing to the whole plane, network window. As the main result, we prove
that such defined (macroscopic) throughput characteristic is equal to the
ratio of the mean traffic demand to the mean number of users in the typical
cell of the network. Both these means account for double averaging: over time
and network geometry.

A key element of the analysis of the cellular network is the spatial
distribution of the signal-to-interference-and-noise ratio (SINR). We show how
this distribution enters into the macroscopic characterization of the
throughput. When considering SINR we are able to account for the fact that the
base stations that are idling, i.e., have no users to serve, do not contribute
to the interference. This makes the performance of different cells
interdependent and we take it into account via a system of cell-load equations.

Finally, we show how to amend the model letting it account for the shadowing
in the path loss. The latter is known to impact the geometry of the network,
in the sense that the serving base station is not necessarily the closest one.
It also alters the distribution of the SINR.

\subsection{Related work}

The evaluation of user QoS metrics in cellular networks is a hard problem, but
crucial for network operators and equipment manufacturers. It motivates a lot
of engineering and research studies. The complexity of this problem made many
actors develop complex and time consuming simulation tools such as those
developed by the industrial contributors to 3GPP (\emph{3rd Generation
Partnership Project})~\cite{3GPP36814-900}. There are many other simulation
tools such as TelematicsLab LTE-Sim~\cite{Piro2011}, University of Vien LTE
simulator~\cite{Mehlfuhrer2011,Simko2012} and LENA
tool~\cite{Baldo2011,Baldo2012} of CTTC, which are not necessarily compliant
with 3GPP.

A possible analytical approach to this problem is based on the information
theoretic
characterization of the individual link performance; cf
e.g.~\cite{GoldsmithChua1997,Mogensen2007}, in conjunction with a queueing
theoretic modeling and analysis of the user traffic;
cf.
e.g.~\cite{Borst2003,BonaldProutiere2003,HegdeAltman2003,BonaldBorstHegdeJP2009,RongElayoubiHaddada2011,KarrayJovanovic2013Load}%
. All these works consider some particular aspects of the network and none of
them considers a large, irregular multi-cell network. Such a scenario is
studied in our approach by using stochastic-geometric tools combined with the
two aforementioned theories. As a result, we propose a global, macroscopic
approach to the evaluation of the user QoS metrics in cellular networks, which
we compare and validate with respect to real network measurements.

Stochastic geometry has already been shown to give analytically tractable
models of cellular networks, see e.g.~\cite{ANDREWS2011}. However, to the best
of our knowledge, the prior works in this context usually do not consider any
dynamic user traffic.

\subsection{Paper organization}

We describe our general cellular network model, comprising the geometry of
base stations, path loss model, space-time traffic demand process, and the
service policy in Section~\ref{s.ModelDescription}. In
Section~\ref{s.TypicalMeanCells} we develop two approaches, called
respectively \emph{typical} and \emph{mean cell } approach, allowing to study
the dependence of the mean user throughput in the network and other
macroscopic network characteristics on the traffic demand and other model
parameters. In Section~\ref{s.NumericalResults} we apply these approaches
studying some particular network model and compare the obtained numerical
results to real field measurements.

\section{Model description and local\newline characteristics}

\label{s.ModelDescription} We shall now describe our model.

\subsection{Network model}

\label{ss.nm} We consider locations $\{X_{1},X_{2},\ldots\}$ of base stations
(BS) on the plane $\mathbb{R}^{2}$ as a realisation of a point process, which
we denote by $\Phi$.~\footnote{According to the formalism of the theory of
point processes (cf e.g.~\cite{DaleyVereJones2003}), a point process is a
random measure $\Phi=\sum_{j}\delta_{X_{j}}$, where $\delta_{x}$\ denotes the
Dirac measure at $x$.} We assume that $\Phi$\ is stationary and ergodic with
positive, finite intensity (mean number of BS per unit of surface)
$\lambda_{\mathrm{BS}}$.~\footnote{Stationarity means that the distribution of
the process is translation invariant, while ergodicity allows to interpret
some mathematical expectations as spatial averages of some network
characteristics.}


In order to simplify the presentation, we shall make first the following two
assumptions, which will be relaxed in Sections~\ref{ss.shadow}
and~\ref{s.PonderedInterference}, respectively.

\begin{enumerate}
\item \label{no-shadow} \emph{There is no shadowing}. The (time-averaged over
fading) propagation loss depends only on the distance $r$ between the emitter
and the receiver through a path-loss function $l(r)$, which we assume increasing.

\item \label{fi} \emph{Full interference.} Each base station is always
transmitting at some fixed power $P$, common for all stations.
\end{enumerate}

We will also assume throughout the whole paper that \emph{each user is served
by the BS which he or she receives with the strongest signal power}. The
consequence of the assumption~\ref{no-shadow} above is that each BS $X\in\Phi
$\ serves users in a geographic zone
$V\left(  X\right)  =\{y\in\mathbb{R}^{2}:|y-X|\leq\min_{Z\in\Phi}|y-Z|\}$
which is called \emph{Voronoi cell} of~$X$ in~$\Phi$.
Both the above assumptions~\ref{no-shadow} and~\ref{fi} allow to represent the
\emph{signal to interference and noise ratio} (SINR) at location $y\in V(X)$,
$X\in\Phi$ as
\begin{equation}
\mathrm{SINR}\left(  y,\Phi\right)  :=\frac{P/l\left(  \left\vert
y-X\right\vert \right)  }{N+P\sum_{Z\in\Phi\backslash\left\{  X\right\}
}1/l\left(  \left\vert y-Z\right\vert \right)  }\,, \label{e.SINR}%
\end{equation}
where $N$ is the noise power. Note that this represents the SINR in the
\emph{down-link} (BS to user) channel. Note also, that we assume that each
interfering base station always transmits (even if, for example, it has no
user to serve). This model, which we call \emph{full interference model}, will
be improved in Section~\ref{s.PonderedInterference}.

We assume that the \emph{peak bit-rate} at location $y$, defined as the
bit-rate of a user located at $y$ when served alone by its BS, is some
function $R\left(  \mathrm{SINR}\right)  $, of the
SINR.~\footnote{\label{fn.peak-rate} The theoretical upper-bound of such
function characterizing the link level performance is typically given by
information theory. For example, in the case of AWGN channel, $R\left(
\mathrm{SINR}\right)  =W\log_{2}\left(  1+\mathrm{SINR}\right)  $\ where
$W$\ is the bandwidth. We shall assume that the \emph{ fading} is already
averaged out by considering the so-called ergodic capacity. For example, in
the case of flat fading, $R\left(  \mathrm{SINR}\right)  =WE\left[  \log
_{2}\left(  1+\left\vert H\right\vert ^{2}\mathrm{SINR}\right)  \right]
$\ where $H$\ is a random variable representing the fading and $E\left[
\cdot\right]  $\ is the expectation with respect to~$H$.}


We consider variable bit-rate (VBR) traffic; i.e., users arrive to the network
and require to transmit some volume of data at a bit-rate decided by the
network.
Specifically, we assume that user channels are \emph{intra-cell orthogonal and
inter-cell independent}: if BS $X$ serves $n$ users located at $y_{1}%
,y_{2},\ldots,y_{n}\in V\left(  X\right)  $\ then the bit-rate of the user
located at $y_{j}$\ equals to $1/n\,$th of its peak bit-rate $\frac{1}%
{n}R\left(  \mathrm{SINR}\left(  y_{j},\Phi\right)  \right)  $, $j\in\left\{
1,2,\ldots,n\right\}  $.~\footnote{This can be achieved using various multiple
access schemes, e.g. time division.}

The pattern of BS $\Phi$ does not evolve in time. We describe now the
space-time process of user arrivals (and departures).

\subsection{Traffic demand}

Users arrive uniformly on the plane and require to transmit a random
(arbitrarily distributed) volume of data of mean $1/\mu$ bits. The duration
between the arrivals of two successive users in each geographic zone $S$ of
surface $|S|$ is an exponential random variable of parameter $\lambda
\times|S|$ . This means that on average there are $\lambda$ arrivals per
surface unit. The arrival locations, inter-arrival durations as well as the
data volumes are assumed independent across users. The above description
corresponds to a \emph{space-time Poisson process of arrivals}, independently
marked by data volumes.

We assume that the users don't move considerably during their
calls.~\footnote{Small user movements are reflected in channel fading; cf. the
remark in footnote~\ref{fn.peak-rate}.} Each user stays in the system for the
time necessary to download his data. This takes a random (service) time
because the bit-rate with which he is served depends on the configuration of
other users \emph{served by the same base station}. Users depart from the
system immediately after having downloaded their data.

The traffic demand per surface unit is then equal to
$\rho=\frac{\lambda}{\mu}$,
which will be expressed in bit/s/km$^{2}$.
The traffic demand in a given cell equals
\begin{equation}
\rho\left(  X\right)  =\rho\left\vert V\left(  X\right)  \right\vert ,\quad
X\in\Phi. \label{e.TrafficDemand}%
\end{equation}

\subsection{Local quality of service characteristics}

\label{s.QoS} For a fixed configuration of BS $\Phi$, the service of users
arriving to the cell $V(X)$ of a given BS $X\in\Phi$ can be modeled by an
appropriate (spatial) multi-class processor sharing queue, with classes
corresponding to different peak bit-rates characterized by user locations
$y\in V(X)$. Note also that a consequence of our model assumptions (in
particular the full interference assumption~\ref{fi}, inter-cell channel
independence and space-time Poisson arrivals) the service processes of
different queues are independent.

We consider now the steady-state of users served in each cell $V(X)$%
.~\footnote{Note that the (mean) QoS characteristics of users in this state
correspond to time-averages of user characteristics.} The following
expressions follow from the queueing-theoretic analysis of the processor
sharing systems of each BS $X\in\Phi$, cf~\cite{BonaldProutiere2003,KarrayJovanovic2013Load} for the details.

\begin{itemize}
\item The service process of BS $X\in\Phi$ is stable if and only if its
traffic demand does not exceed the critical value that is the harmonic mean of
the peak bit-rate over the cell:
\begin{equation}
\rho_{\mathrm{c}}\left(  X\right)  :=\frac{\left\vert V\left(  X\right)
\right\vert }{\int_{V\left(  X\right)  }1/R\left(  \mathrm{SINR}\left(
y,\Phi\right)  \right)  dy}\,. \label{e.CriticalTraffic}%
\end{equation}
Note that $\rho_{c}(X)$ depends on $X$ and on $\Phi$. The same observation is
valid for the subsequent cell characteristics.

\item The mean user throughput in the given cell, defined as the ratio of the
mean volume of the data request $1/\mu$ to the average service time of users
\emph{in this cell}, can be expressed as follows
\begin{equation}
r\left(  X\right)  =\max(\rho_{\mathrm{c}}\left(  X\right)  -\rho\left(
X\right)  ,0)\,. \label{e.UserThroughput}%
\end{equation}

\item The mean number of users in the steady state of the given cell equals
to
\begin{equation}
N\left(  X\right)  =\frac{\rho\left(  X\right)  }{r\left(  X\right)  }\,.
\label{e.UsersNumber}%
\end{equation}
Note that $N(X)=\infty$ if $\rho(X)\ge\rho_{c}(X)$.

\item The probability that the given BS is not idling in the steady state (has
at least one user to serve) equals
\begin{equation}
p\left(  X\right)  =\min\left(  \theta\left(  X\right)  ,1\right)  \,,
\label{e.Proba}%
\end{equation}
where $\theta(X)$, which we call \emph{cell load}, is defined as
\begin{equation}
\theta\left(  X\right)  :=\frac{\rho\left(  X\right)  }{\rho_{\mathrm{c}%
}\left(  X\right)  }\,. \label{e.Load1}%
\end{equation}

\end{itemize}

Note that the cell is stable if and only if $\theta(X)<1$ and
\begin{equation}
\label{e.theta-R}\theta(X)=\rho\int_{V\left(  X\right)  }1/R\left(
\mathrm{SINR}\left(  y,\Phi\right)  \right)  dy\,.
\end{equation}
Moreover,
\begin{align}
N(X)  &  =\frac{\theta(X)}{1-\theta(X)}\,,\label{e.N-theta}\\
r(X)  &  =\rho(X)(1/\theta(X)-1) \label{e.r-theta}%
\end{align}
provided $\theta(X)<1$.

The above expressions allow to express all other characteristics in terms of
the traffic demand per cell $\rho(X)$ and the cell load $\theta(X)$.

\begin{remark}
All the above characteristics are \emph{local} network characteristics in the
sense that they characterize the service at each BS $X$ and vary over
$X\in\Phi$. Real data analysis and simulations for Poisson network models
exhibit a lot of variability among these characteristics. In particular,
plotting the mean user throughput $r(X)$ as function of the mean traffic
demand $\rho(X)$ for different $X\in\Phi$ does not reveal any apparent
systematic relation between these two local characteristics; cf.
Figure~\ref{f.crude-data}.
\end{remark}

\section{Global network characteristics}

\label{s.TypicalMeanCells} In this section we propose some global
characteristics of the network allowing to characterize its macroscopic
performance. We are particularly interested in finding such a relation between
the (per surface) traffic demand $\rho$ and the (global) mean user throughput
in the network, with this latter characteristic yet to be properly defined.

\subsection{Typical cell of the network}

\label{ss.typical-cell} A first, natural idea in this regard is to consider
spatial averages of the local characteristics in an increasing network window
$A$, say a ball centered at the origin and the radius increasing to infinity.
Assuming ergodicity of the point process $\Phi$ of the BS, these averages can
be expressed and calculated as Palm-expectations of the respective
characteristics of the so called \textquotedblleft typical
cell\textquotedblright\ $V(0)$. For example
\begin{equation}
\lim_{|A|\rightarrow\infty}1/\Phi(A)\sum_{X\in A}\rho(X)=\mathbf{E}^{0}%
[\rho(0)]=\rho\mathbf{E}^{0}[|V(0)|]\,. \label{e.rho-ergodicity}%
\end{equation}
The typical cell $V(0)$ is the cell of the BS located at the origin $X=0$ and
being part of the network $\Phi$ distributed according to the \emph{Palm
distribution} $\Pr^{0}$ associated to the original stationary distribution
$\Pr$ of $\Phi$. In the case of Poisson process, the relation between the Palm
and stationary distribution is particularly simple and (according to
Slivnyak's theorem) consists just in adding the point $X=0$ to the stationary
pattern $\Phi$.

The convergence analogue to~(\ref{e.rho-ergodicity}) holds for each of the
previously considered local characteristics $\mathbf{E}^{0}[\rho_{c}(0)]$,
$\mathbf{E}^{0}[r(0)]$, $\mathbf{E}^{0}[N(0)]$, $\mathbf{E}^{0}[p(0)]$ and
$\mathbf{E}^{0}[\theta(0)]$. The convergence is $\Pr$~almost sure and follows
from the ergodic theorem for point processes (see~\cite[Theorem~4.2.1]%
{BaccelliBlaszczyszyn2009T1},~\cite[Theorem~13.4.III]{DaleyVereJones2003}).
However, as we will explain in what follows, \emph{not all} of these
\emph{mean-typical cell} characteristics have natural interpretations as
macroscopic network characteristics.

First, note that the existence of some (even arbitrarily small) fraction of BS
$X$ which are not stable (with $\rho(X)\ge\rho_{c}(X)$, hence $N(X)=\infty$)
makes $\mathbf{E}^{0}[N(0)]=\infty$.

\begin{remark}
\label{r.nonstable-cells} For a well dimensioned network one does not expect
unstable cells. For a perfect hexagonal network model $\Phi$ \emph{all} cells
are stable or unstable depending on the value of the per-surface traffic
demand $\rho$. An artifact of an infinite, homogeneous, Poisson model $\Phi$
is that for arbitrarily small $\rho$ there exists a non-zero fraction of BS
$X\in\Phi$, which are non-stable. This fraction is very small for reasonable
$\rho$, allowing to use Poisson to study QoS metrics which, unlike
$\mathbf{E}^{0}[N(0)]$, are not ``sensitive'' to this artifact.
\end{remark}

We will also show in the next section that it is \emph{not} natural to
interpret $\mathbf{E}^{0}[r(0)]$ (which is \emph{not} sensitive to the
existence of a small fraction of unstable cells) as the mean user throughput
in the network; see Remark~\ref{r.network-throuhput}. Before we give an
alternative definition of this latter QoS, let us state the following result,
which will be useful in what follows.

\begin{proposition}
\label{p.rho-theta} We have
\begin{align}
\label{e.rho0-exchange}\mathbf{E}^{0}[\rho(0)]  &  =\frac{\rho}{\lambda
_{\mathrm{BS}}}\,,\\
\mathbf{E}^{0}[\theta(0)]  &  =\frac{\rho}{\lambda_{\mathrm{BS}}}%
\mathbf{E}[1/R\left(  \mathrm{SINR}\left(  0,\Phi\right)  \right)  ]\,.
\label{e.TypicalLoad}%
\end{align}

\end{proposition}

\begin{proof}
The first equation is quite intuitive: the average cell surface is equal to
the inverse of the average number of BS per unit of surface. Formally, both
equations follow from the inverse formula of Palm
calculus~\cite[Theorem~4.2.1]{BaccelliBlaszczyszyn2009T1}. In particular,
for~(\ref{e.TypicalLoad}) one uses representation~(\ref{e.theta-R}) in
conjunction with the inverse formula.
\end{proof}

\begin{remark}
\label{r.sinr} Note that the expectation in the right-hand-side
of~(\ref{e.TypicalLoad}) is taken with respect to the \emph{stationary}
distribution of the BS process $\Phi$. It corresponds to the spatial average
of the inverse of the peak bit-rate calculated throughout the network.
The (only) random variable in this expression is the SINR experienced by the
\emph{typical user}.
This distribution is usually known in operational networks (estimated from
user measurements). It can be also well approximated using Poisson network
model for which its distribution function admits an explicit expression;
cf~\cite{BlaszczyszynKarrayKeeler2013,coverage}.
\end{remark}

\subsection{Mean user throughput in the network}

\label{ss.throughput-network} Faithful to the usual definition of the mean
user throughput as the ratio of the mean volume of the data request to the
mean service duration (which we retained at the local, cell level) we aim to
define now the mean user throughput in the (whole) network as the ratio of
these two quantities taken for increasing network window $A$. However in order
to \textquotedblleft filter out\textquotedblright\ the impact of cells which
are not stable and avoid undesired degeneration of this characteristic (e.g.
for Poisson process; cf. Remark~\ref{r.nonstable-cells})
let us consider the union of all stable cells
\[
\mathcal{S}:=\bigcup_{X\in\Phi:\rho(X)<\rho_{c}(X)}V(X)\,.
\]
Note that the stationarity of $\Phi$ implies the same for the random set
$\mathcal{S}$. We denote by $\pi_{\mathcal{S}}=\mathbf{E}[\mathbf{1}%
(0\in\mathcal{S})]$ the volume fraction of $\mathcal{S}$ and call it the
\emph{stable fraction of the network}. It is equal to the average fraction of
the plane covered by the stable cells; cf.~\cite[Definition 3.4 and the
subsequent Remark]{BaccelliBlaszczyszyn2009T1}. Denote also
\begin{equation}
N^{0}:=\mathbf{E}^{0}[N(0)\mathbf{1}(N(0)<\infty)]\,.
\end{equation}
We are ready now to define the \emph{mean user throughput in the network}
$r^{0}$ as the ratio of the average number of bits per data request to the
average duration of the data transfer in the stable part of the network
\begin{equation}
r^{0}:=\lim_{|A|\rightarrow\infty}\frac{1/\mu}{\text{(temporal-)mean service
time in $A\cap\mathcal{S}$}}\,. \label{e.r0-def}%
\end{equation}
Here is the key result of the typical cell approach. Its proof is given at the
end of this section.

\begin{theorem}
\label{t.r-via-Little} For ergodic network $\Phi$ we have
\begin{equation}
\label{e.r-via-Little}r^{0}=\frac{\rho\,\pi_{\mathcal{S}}}{\lambda
_{\mathrm{BS}} N^{0}}\,.
\end{equation}

\end{theorem}

\begin{remark}
Equation~(\ref{e.r-via-Little}) provides a macroscopic relation between the
traffic demand and the mean user throughput in the network, which we are
primarily looking for in this paper. It will be validated by comparison to
real data measurements. Quantities $N^{0}$ and $\pi_{\mathcal{S}}$ do not have
explicit analytic expressions analogous to~(\ref{e.TypicalLoad}). Nevertheless
they can be estimated from simulations of a given network model $\Phi$. Note
that these are static simulations of the network model. No simulation of the
traffic demand process is necessary, which greatly simplifies the task. For
small and moderate values of the traffic demand (observed in real networks)
one obtains $\pi_{\mathcal{S}}\simeq1$. Moreover, in
Section~\ref{ss.mean-cell} we will propose some more explicit approximation of
$N^{0}$.

\end{remark}

\begin{remark}
\label{r.network-throuhput} Assume that there are no unstable cells in the
network. This is the case e.g. for lattice (say hexagonal) network models with
traffic demand $\rho<\rho_{c}(X)=\rho_{c}$, where the value of the critical
traffic is the same for all cells. Then $\pi_{\mathcal{S}}=1$, $N^{0}%
=\mathbf{E}^{0}[N(0)]$ and the relation~(\ref{e.r-via-Little})
takes form
\begin{equation}
\label{e.r-via-Little-stable}r^{0}=\frac{\rho}{\lambda_{\mathrm{BS}}
\mathbf{E}^{0}[N(0)]}=\frac{\mathbf{E}^{0}[\rho(0)]}{ \mathbf{E}^{0}[N(0)]}\,.
\end{equation}
Thus, in general $r^{0}\not =\mathbf{E}^{0}[r(0)]=\mathbf{E}^{0}%
[\rho(0)/N(0)]$. We want to emphasize that this is not merely a theoretical
detail resulting from our (and common) definition of the mean
throughput~(\ref{e.r0-def}). The expression $\mathbf{E}^{0}[r(0)]=\mathbf{E}%
^{0}[\rho(0)/N(0)]$, which in principle can be considered as another global
QoS metric, is in practice difficult to estimate.
Indeed, when estimating $\mathbf{E}^{0}[r(0)]$ as the average of the ratio
``traffic demand to the number of users'' from real data measurements, one
needs to give a special treatment to observations which correspond to cells
during their idling hours (i.e., with no user, and such observations are not
rare in operational networks). Neither skipping nor literal acceptance of
these observations captures the right dependence of the mean user throughput
on the traffic demand.
\end{remark}

\begin{proof}[Proof of Theorem~\ref{t.r-via-Little}]By Little's law~\cite[eq. (3.1.14)]%
{baccelli2003elements} the temporal mean service time $T^{W}$ of users in any
region of the network $W$, say the union of stable cells with BS in some
region $A$, $W=\bigcup_{X\in A\cap\mathcal{S}}V(X)$, is related to the mean
number $N^{W}$ of the users served in this region $W$ in the steady state by
the equation $N^{W}=\lambda|W|T^{W}$. Consequently, the mean user throughput
in this region $W$ can be expressed as $1/(\mu T^{W})=\rho|W|/N^{W}$. Using
\begin{align}
\frac{|W|}{N^{W}} &  =\frac{|W|}{\sum_{X\in A\cap\mathcal{S}}N(X)}\\[1ex]
&  =\frac{\sum_{X\in A}|V(X)|\mathbf{1}(N(X)<\infty)}{|A|}\frac{|A|}%
{\sum_{X\in A\cap\mathcal{S}}N(X)}\,
\end{align}
and again the ergodic theorem for point yprocess $\Phi$, we obtain the that
limit in~(\ref{e.r0-def}) is $\Pr$-almost surely equal to $\rho\mathbf{E}%
^{0}[|V(0)|\mathbf{1}(N(0)<\infty)]/\mathbf{E}^{0}[N(0)\mathbf{1}%
(N(0)<\infty)]$. By the aforementioned inverse formula of Palm calculus we
conclude $\mathbf{E}^{0}[|V(0)|\mathbf{1}(N(0)<\infty)]=\mathbf{E}%
[\mathbf{1}(0\in\mathcal{S})]/\lambda_{\mathrm{BS}}$.
\end{proof}

\subsection{Mean cell}

\label{ss.mean-cell} It is tempting to look for a synthetic model which would
allow to relate main parameters and QoS metrics of a large irregular cellular
network in a simple, yet not simplistic way. The typical cell approach
described up to now offers such possibility. In this section we will go a
little bit further and propose an even simpler model. It consists in
considering a \emph{virtual} cell, to which we will assign the parameters and
QoS metrics inspired by the analysis of the typical cell. In contrast to the
typical cell, our virtual cell is not random and this is why we call it the
\emph{mean cell}. Specifically, we define it as a (virtual) cell having the
same traffic demand $\bar\rho$ and load $\bar\theta$ as the typical cell. Note
that these two characteristics admit explicit expressions; cf.
Proposition~\ref{p.rho-theta} and Remark~\ref{r.sinr}.
\begin{align}
\bar{\rho}  &  :=\mathbf{E}^{0}\left[  \rho\left(  0\right)  \right]
=\frac{\rho}{\lambda_{\mathrm{BS}}}\,,\label{e.TrafficDemand0}\\
\bar{\theta}  &  :=\mathbf{E}^{0}\left[  \theta\left(  0\right)  \right]  =
\frac{\rho}{\lambda_{\mathrm{BS}}}\mathbf{E}[1/R\left(  \mathrm{SINR}\left(
0,\Phi\right)  \right)  ]\,. \label{e.Load0}%
\end{align}
For the remaining characteristics, we assume that they are related to the
above two via the relations presented in Section~\ref{s.QoS}. Specifically,
following~(\ref{e.TrafficDemand}) we define the surface of the mean cell by
$\bar V=\bar\rho/\rho$ and in analogy to~(\ref{e.Load1})
we define the critical load of the mean cell as
\begin{equation}
\bar{\rho}_{\mathrm{c}}:=\frac{\bar{\rho}}{\bar{\theta}}\,.
\end{equation}
We say that the mean cell is stable if $\bar\rho<\bar\rho_{c}$. Inspired
by~(\ref{e.UserThroughput}) we define the user's throughput in the mean cell
by
\[
\bar{r}:=\max\left(  \bar{\rho}_{\mathrm{c}}-\bar{\rho},0\right)
\]
and, as in~(\ref{e.UsersNumber}),
the mean number of users in the mean cell is defined as
\[
\bar{N}:=\frac{\bar{\rho}}{\bar{r}}\,.
\]
We observe the following immediate relations.
\begin{corollary}
The mean cell is stable if and only if $\bar\theta<1$. In this case
\begin{align}
\label{e.barN-theta}\bar N  &  =\frac{\bar\theta}{1-\bar\theta}\,,\\
\bar r  &  =\bar\rho(1/\bar\theta-1)\,, \label{e.barr-theta}%
\end{align}
which are analogous to~(\ref{e.N-theta}) and~(\ref{e.r-theta}), respectively.
\end{corollary}

\begin{remark}
The equation~(\ref{e.barr-theta}) provides an alternative macroscopic relation
between the traffic demand and the mean user throughput in the network. It is
purely analytic; no simulations are required provided one knows the
distribution of the SINR of the typical user in~(\ref{e.Load0}). It will be
validated by comparison to real data measurements. We consider it as an
approximation of~(\ref{e.r-via-Little}). It consists in assuming $\bar N\simeq
N^{0}/\pi_{\mathcal{S}}$. This latter hypothesis will be also separately
validated numerically.
\end{remark}
\begin{remark}
Note that the key characteristic of the mean cell is its load $\bar\theta$. In
analogy to the load factor of the (classical) M/G/1 processor sharing queue,
it characterizes the stability condition, mean number of users and the mean
user throughput.
\end{remark}


\subsection{Cell-load equations}

\label{s.PonderedInterference} We have to revoke now the full interference
assumption~\ref{fi} made in Section~\ref{ss.nm}. An amendment is necessary in
this matter for the model to be able to predict the real network data; cf
numerical examples in Section~\ref{s.NumericalResults}. Recall that the
consequence of this assumption is that in the expression~(\ref{e.SINR}) of the
SINR all the interfering BS are always transmitting at a given power $P$. In
real networks BS transmit only when they serve at least one
user.~\footnote{Analysis of more sophisticated power control schemes is beyond
the scope of this paper.} Taking this fact into account in an exact way
requires introducing in the denominator of~(\ref{e.SINR}) the indicators that
a given station $Z\in\Phi$ at a given time is not idling. This, in
consequence, would lead to the probabilistic dependence of the service process
at different cells and result in a non-tractable
model.~In particular, we are not aware of any result regarding the stability of such a family of dependent queues.
For this reason, we take into account whether $Z$ is idling or not in a
simpler way, multiplying its powers $P$ by the \emph{probability} $p(Z)$ that
it is not idle in the steady state. In other words we modify the expression of
the SINR as follows
\begin{equation}
\mathrm{SINR}\left(  y,\Phi\right)  :=\frac{P/l\left(  \left\vert
y-X\right\vert \right)  }{N+P\sum\limits_{Z\in\Phi\backslash\left\{
X\right\}  }p\left(  Z\right)  /l\left(  \left\vert y-Z\right\vert \right)
},\label{e.LoadInterference}%
\end{equation}
for $y\in V(X)$, $X\in\Phi$ where $p(Z)$ are cell non-idling probabilities
given by~(\ref{e.Proba}). We will see in Section~\ref{s.NumericalResults} that
this model, called \emph{(load-)weighted interference model}, fits better to
real field measurements than the full interference model. 
The above modification of the model
preserves the independence of the processor-sharing queues at different cells
given the realization $\Phi$ of the network (thus allowing for the
explicit analysis of Section~\ref{s.QoS}).  However
the cell loads $\theta(X)$ are no longer functions of
the traffic demand and the SINR experienced in the respective  cells, but are
related to each other by the following equations that replace~(\ref{e.theta-R}%
)
\begin{align}
\theta\left(  X\right) 
  =\rho\int_{V\left(  X\right)  }\frac{1}{R\left(  \frac{P/l\left(
\left\vert y-X\right\vert \right)  }{N+P\sum\limits_{Z\in\Phi\backslash
\left\{  X\right\}  }\min\left(  \theta\left(  Z\right)  ,1\right)  /l\left(
\left\vert y-Z\right\vert \right)  }\right)  }\,dy\,.\label{e.LoadSystem}%
\end{align}
We call this system of equations in the unknown cell loads $\left\{
\theta\left(  X\right)  \right\}  _{X\in\Phi}$\ the \emph{cell-load equations}.

\begin{remark}[Spatial stability]
The weighted interference model
introduces more ``spatial'' dependence between the processor sharing
queues of different cells, while preserving their ``temporal''
(conditionally, given $\Phi$)  independence.
A natural question regarding the existence and  uniqueness of the
solution of the fixed point problem~(\ref{e.LoadSystem}) arises.
Note that the mapping in the right-hand-side of~(\ref{e.LoadSystem})
is increasing  in all $\theta(Z)$, $Z\in\Phi$ provided function $R$ is increasing.
Using this property it is easy to see that successive iterations of this mapping
started off $\theta(Z)\equiv 0$ on one hand side and off $\theta(Z)$ as in~(\ref{e.theta-R}) (full interference model) on the other side,
converge to a minimal and maximal solution of~(\ref{e.LoadSystem}),
respectively.
An interesting theoretical question  regards the uniqueness
of the solution of~(\ref{e.LoadSystem}), in particular  for a random, say Poisson,  point process
$\Phi$.  Answering this question, which we call ``spatial stability'' of the model, is
unfortunately beyond the scope of this paper.%
\footnote{Existence and uniqueness of the solution of a very similar
  problem (with finite number of stations and a discrete traffic demand)
  is proved in~\cite{siomina2012analysis}.}
The simulation study of the typical cell
model,  presented in Section~\ref{s.NumericalResults} (where we use Matlab to find a solution of (\ref{e.LoadSystem}) for any given
finite pattern of base stations $\Phi$) is less stable for larger
values of the traffic demand $\rho$.
\end{remark}

In the mean cell approach (cf Section~\ref{ss.mean-cell}) we take into account
the weighted interference model by the following (single) equation in the
mean-cell load $\bar\theta$
\begin{align}
\bar{\theta}  &  =\frac{\rho}{\lambda_{\mathrm{BS}}}\mathbf{E}\left[
1/R\left(  \frac{P/l\left(  \left\vert X^{\ast}\right\vert \right)  }%
{N+P\sum_{Z\in\Phi\backslash\left\{  X^{\ast}\right\}  }\bar{\theta}/l\left(
\left\vert y-Z\right\vert \right)  }\right)  \right]  \,,
\label{e.LoadEquation}%
\end{align}
where $X^{\ast}$\ is the location of the BS whose cell covers the origin. We
solve the above equation with $\bar{\theta}$ as unknown. We will show in the
numerical section that the solution of this equation\ gives a good estimate of
the empirical average of the loads $\left\{  \theta\left(  X\right)  \right\}
_{X\in\Phi}$\ obtained by solving the system of cell-load
equations~(\ref{e.LoadSystem}) for the simulated model.

\begin{remark}
[Pilot channel] The cells which are not idle might still emit some power (e.g.
in the pilot channel). This can be taken into account by replacing
$p(Z)=\min(\theta(Z),1)$ in~(\ref{e.LoadSystem}) by $p(Z)(1-\epsilon
)+\epsilon$, where $\epsilon$ is the fraction of the power emitted all the
time. Similar modification concerns $\bar{\theta}$ in the right-hand-side
of~(\ref{e.LoadEquation}).
\end{remark}

\subsection{Shadowing}

\label{ss.shadow} Until now we were assuming that the propagation loss is only
induced by the distance between the transmitter and the receiver. In this
section we will briefly explain how the effect of shadowing can be taken into account.

Assume that the shadowing between a given station $X\in\Phi$ and all locations
$y\in\mathbb{R}^{2}$ is modeled by some random field $S_{X}\left(  y-X\right)
$. That is, we assume the propagation loss between $X$ and $y$
$
L_{X}\left(  y\right)  =\frac{l\left(  \left\vert y-X\right\vert \right)
}{S_{X}\left(  y-X\right)  }$.
We assume that, given $\Phi$, the random fields $S_{X}(\cdot)$ are independent
across $X\in\Phi$ and identically distributed. In general, we do not need to
assume any particular distribution for $S_{X}(\cdot)$ (neither independence
nor the same distribution of $S_{X}(y)$ across $y$).

The assumption that each user is served by the BS received with the smallest
path-loss results in the following modification of the geographic service zone
of $X$, which we keep calling ``cell''%
\begin{equation}
\label{e.V-shadow}V\left(  X\right)  =\{y\in\mathbb{R}^{2}:L_{X}\left(
y\right)  \leq\min_{Z\in\Phi}L_{Z}\left(  y\right)  \}\,.
\end{equation}
For mathematical consistence we shall assume that, almost surely, the origin
belongs to a unique cell (i.e., is not located on any cell boundary).

The SINR at location $y$ can be expressed by~(\ref{e.SINR})
or~(\ref{e.LoadInterference}) with $l(|y-Y|)$ replaced by $L_{Y}(y)$ for
$Y\in\Phi$, depending whether we consider the full interference or the
weighted interference model, for $y\in V(X)$, with $V(X)$ defined
by~(\ref{e.V-shadow}). The same modification regards the cell-load
equations~(\ref{e.LoadSystem}) and~(\ref{e.LoadEquation}).

All the previous results involving the typical cell remain valid for this
modification of the model. In particular, the results of
Proposition~\ref{p.rho-theta} can be extended to the model with shadowing
(where the cell associated to each base station is not necessarily the Voronoi
cell) provided the origin $0$ belongs to a unique cell almost surely.

Note that the mean cell surface $\mathbf{E}^{0}[|V(0)|]=1/\lambda
_{\mathrm{BS}}$, and hence the mean traffic demand per cell $\mathbf{E}%
^{0}[\rho(0)]=\rho/\lambda_{\mathrm{BS}}$, do not depend at all on the
shadowing. The values of other characteristics of the typical and the mean
cell will change depending on the distribution of the random shadowing field
$S_{X}(y)$ (both the marginal distributions and the correlation across $y$).
An interesting remark in this regard is as follows.

In the full interference model, the mean load of the typical cell and the load
of the mean cell (which are by the definition equal $\mathbf{E}^{0}%
[\theta(0)]=\bar\theta$) depend only on the stationary marginal distribution
of $\text{SINR}(0,\Phi)$, cf.~(\ref{e.Load0}). Hence, it does not depend on
the (spatial) correlation of $S_{X}(y)$ across $y$. Moreover, this
distribution is known in the case of the Poisson network and identically
distributed marginal shadowing $S_{X}(y)\sim S$. As explained
in~\cite{coverage}, in this case $\text{SINR}(0,\Phi)$ in the model with
shadowing has the same distribution as in the model without shadowing and the
density of stations equal to $\lambda_{\mathrm{BS}}\times\mathbf{E}%
[S^{2/\beta}]$. In particular, a specific distribution of $S$ and the
correlation of $S_{X}(y)$ across $y$ play no role. This equivalence of the two
models (with and without shadowing) is more general, as explained
in~\cite{coverage,netequivalence}, and applies also to the mean-cell-load
equation~(\ref{e.LoadEquation}).

\begin{remark}
The impact of the shadowing on the \emph{mean cell model, both in full and
weighted interference scenario} in the above Poisson model with $S_{X}(y)\sim
S$ can be summarized as follows. It modifies only the cell load $\bar\theta$
and not the traffic demand $\bar\rho$. Moreover, multiplying $\lambda
_{\mathrm{BS}}$ by $\mathbf{E}[S^{2/\beta}]$ and dividing $\rho$ by the same
moment one obtains an equivalent (in terms of all considered characteristics)
mean cell without shadowing.
\end{remark}

\section{Numerical results}
\label{s.NumericalResults}
\begin{figure}[t]
\begin{center}
\includegraphics[width=0.9\linewidth]{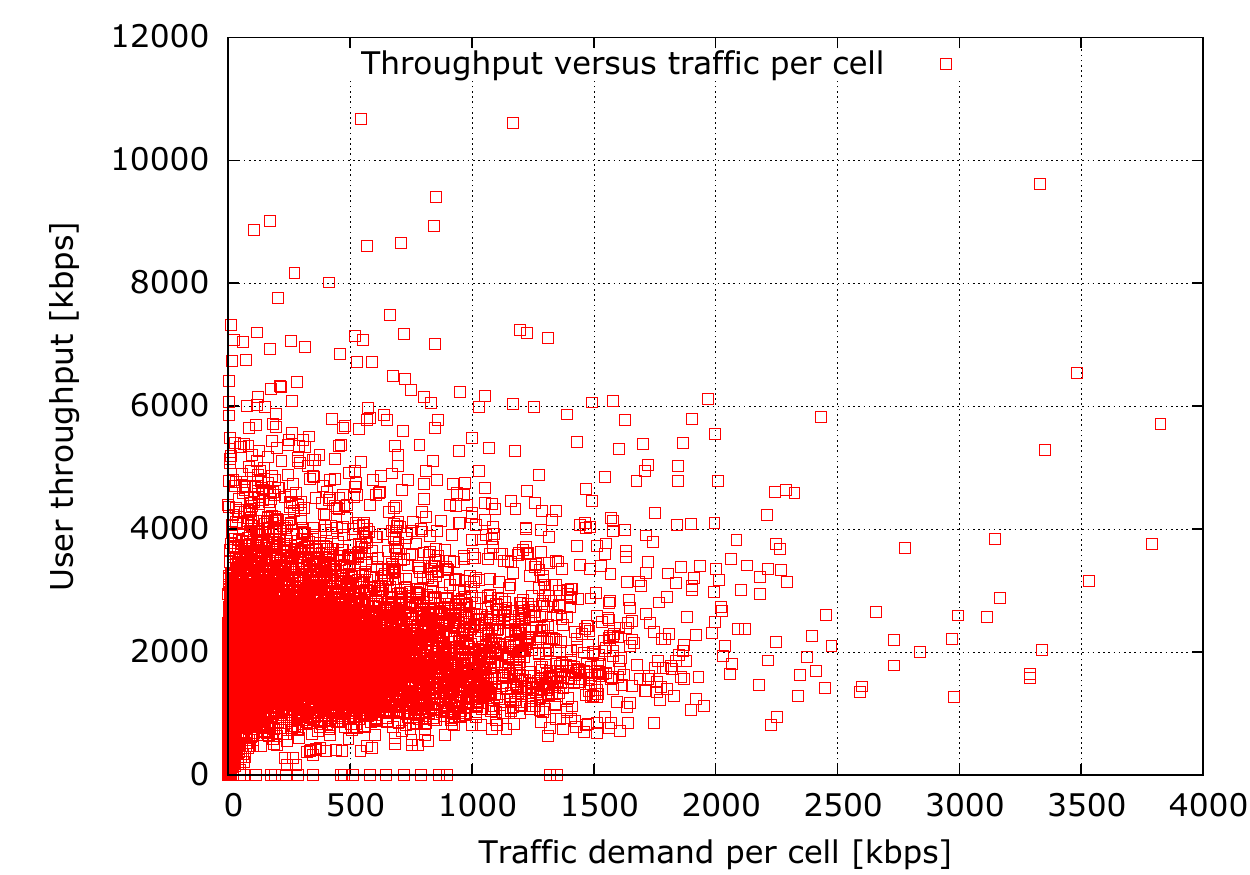}
\end{center}
\caption{Local user throughput versus local traffic demand for some
zone (selected to satisfy  a spatial homogeneity of the base stations)
of an operational cellular network deployed in a big city in
Europe. 9288 different points correspond to the measurements made by
different sectors of different base stations during 24 different hours of
some given day.}%
\vspace{-2ex}
\label{f.crude-data}%
\end{figure}

To illustrate the motivation of this work, we present first on
Figure~\ref{f.crude-data} non-averaged data obtained from the measurements
performed in an operational network in some zone of some big city in Europe.
\footnote{More precisely, a {\em dense urban network zone} consisting of 382 base stations was selected
in a big European city, whose locations loosely satisfy homogeneous spatial Poisson assumption.
Ripley's $L$-function, cf~\cite[page~50]{stoyan1995stochastic}, plotted on
Figure~\ref{f.Ripley}, was used to verify this latter assumption.
The density of base stations in this dense
urban zone  is about 4.62  base stations per km${}^2$. Later, we will consider also a {\em urban zone} of a different European
city, where the spatial homogeneous Poissonianity of the base station
locations can also be retained; cf. Figure~\ref{f.Ripley},
with roughly four times smaller density of base stations,
more precisely 1.15 stations per km${}^2$. In both cases the network operates HSDPA system with MMSE coding.}%

Different points in this figure correspond to the measurements of the traffic
demand and the estimation of the user throughput made by different cells
during different hours of the day. No apparent relation between these two quantities
can be observed in this way.

In order to uderstand and predict the performance of the network for which we
have presented the above data, we will now specify correspondingly our general
model and study it using the proposed approach. The obtained results will be
compared to the appropriately averaged real field measurements.

\subsection{Model specification}

Consider the following numerical setup. Assume Poisson process of BS with
intenisty $\lambda_{\mathrm{BS}}=4.62$km$^{-2}$ (which corresponds to an
average distance between two neighbouring BS of $0.5$km). We assume the
path-loss function $l(r)=(K r)^{\beta}$, with $K=7117$km$^{-1}$, and the path
loss exponent $\beta=3.8$. The propagation model comprises the log-normal
shadowing with the logarithmic standard deviation $10$dB;
cf~\cite{BlaszczyszynKarrayKeeler2013}, and the mean spatial correlation
distance $0.05$km.

The transmision power is $P=58$dBm, with the fraction $\epsilon=10\%$ used in
the pilot channel. The antenna pattern is described in~\cite[Table
A.2.1.1-2]{3GPP36814-900}. The noise power is \hbox{$N=-96$dBm}.

We assume the peak bit-rate equal to 30\% of the ergodic capacity of the AWGN
channel; cf. Footnote~\ref{fn.peak-rate}, with the frequency bandwidth
$W=5$MHz and the Rayleigh fading with mean power $E[|H|^{2}]=1$.

Estimations of the typical cell are performed by the simulation of 30
realizations of the Poisson model within a finite observation window, which is
taken to be the disc of radius $2.63$km. We first average over all BS in this
window and then over the model realizations.
The empirical standard deviation
form the obtained averages will be  presented via error-bars.

We shall study now our model using the typical and mean cell approach,
assuming first the full interference model and then the weighted one.

\subsection{Full interference}
\begin{figure}[t!]
\begin{center}
\includegraphics[width=0.8\linewidth]
{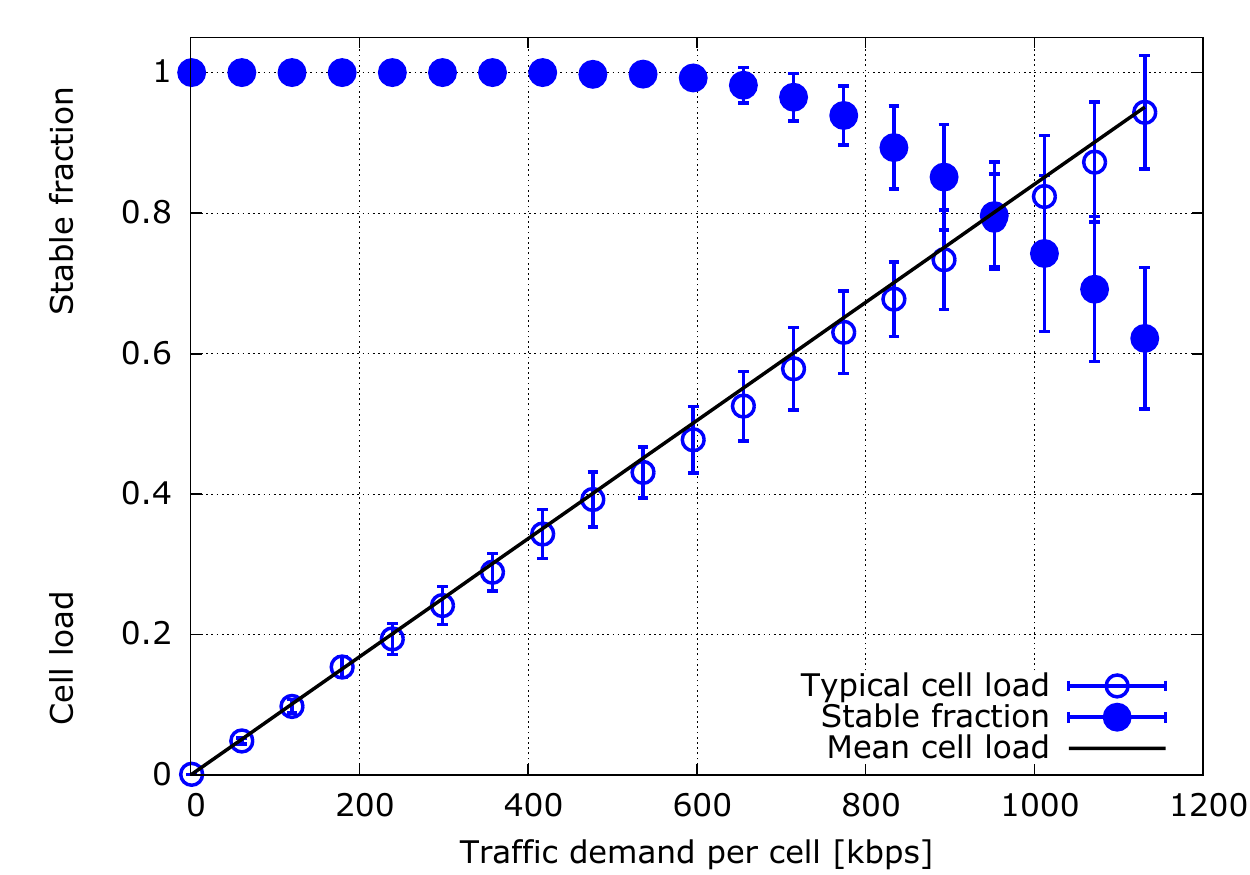}%
\caption{Cell load and the stable fraction of the network versus traffic demand per cell in the full  interference model.}%
\label{f.FullInterference_Load}%
\end{center}
%
\begin{center}
\includegraphics[width=0.8\linewidth]
{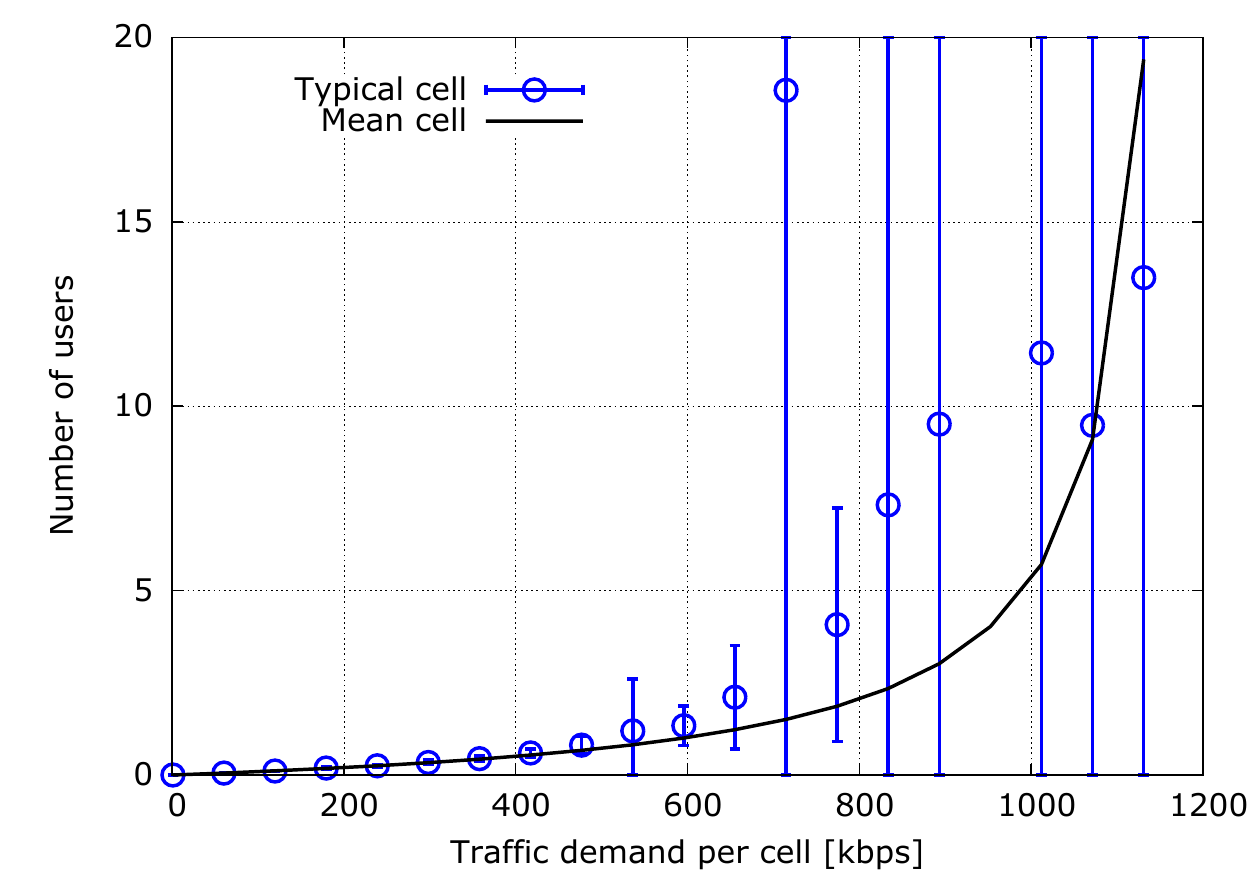}%
\caption{Number of users per cell
 versus traffic demand per cell in the full interference model.}
\label{f.FullInterference_UsersNumber}%
\end{center}
%
\begin{center}
\includegraphics[width=0.8\linewidth]
{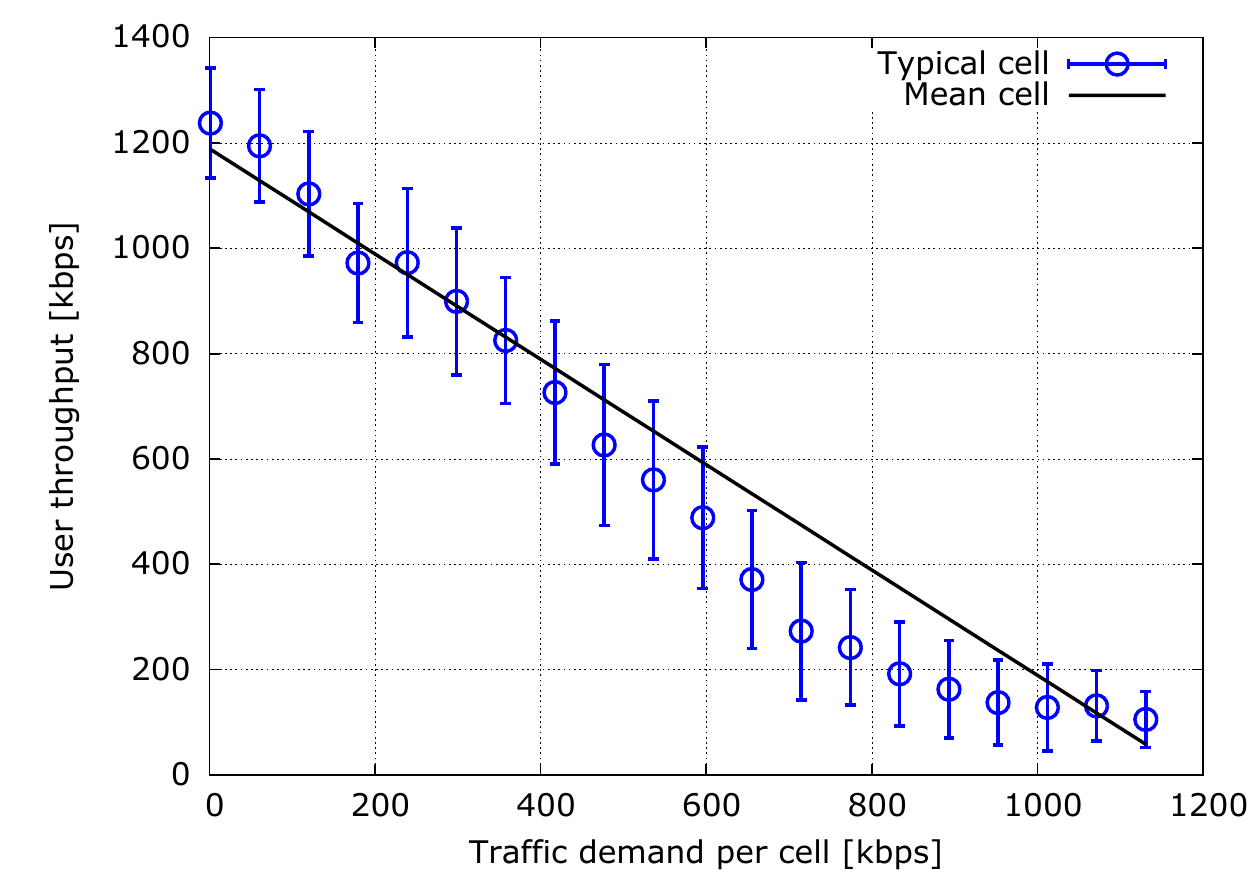}%
\caption{Mean user throughput in the network
 versus traffic demand per cell in the full interference model.}
\label{f.FullInterference_UserThroughput}%
\vspace{-4ex}
\end{center}
\end{figure}
%

\begin{figure}[t!]
\begin{center}
\includegraphics[width=0.8\linewidth]
{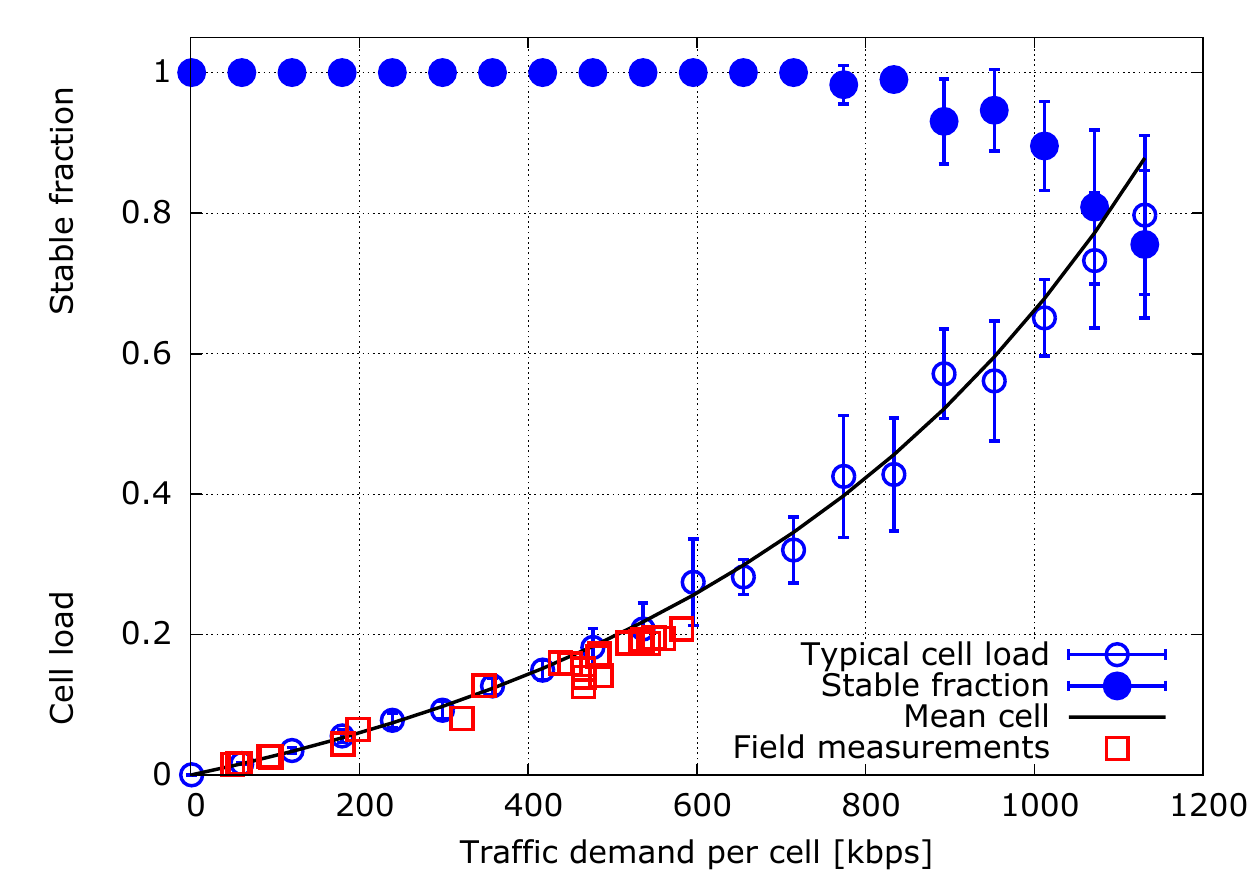}%
\caption{Load and the stable fraction of the network versus traffic demand in the weighted  interference model. Also, load estimated from real field measurements.}%
\label{f.PonderedInterference_Load}%
\end{center}
%
\begin{center}
\includegraphics[width=0.8\linewidth]
{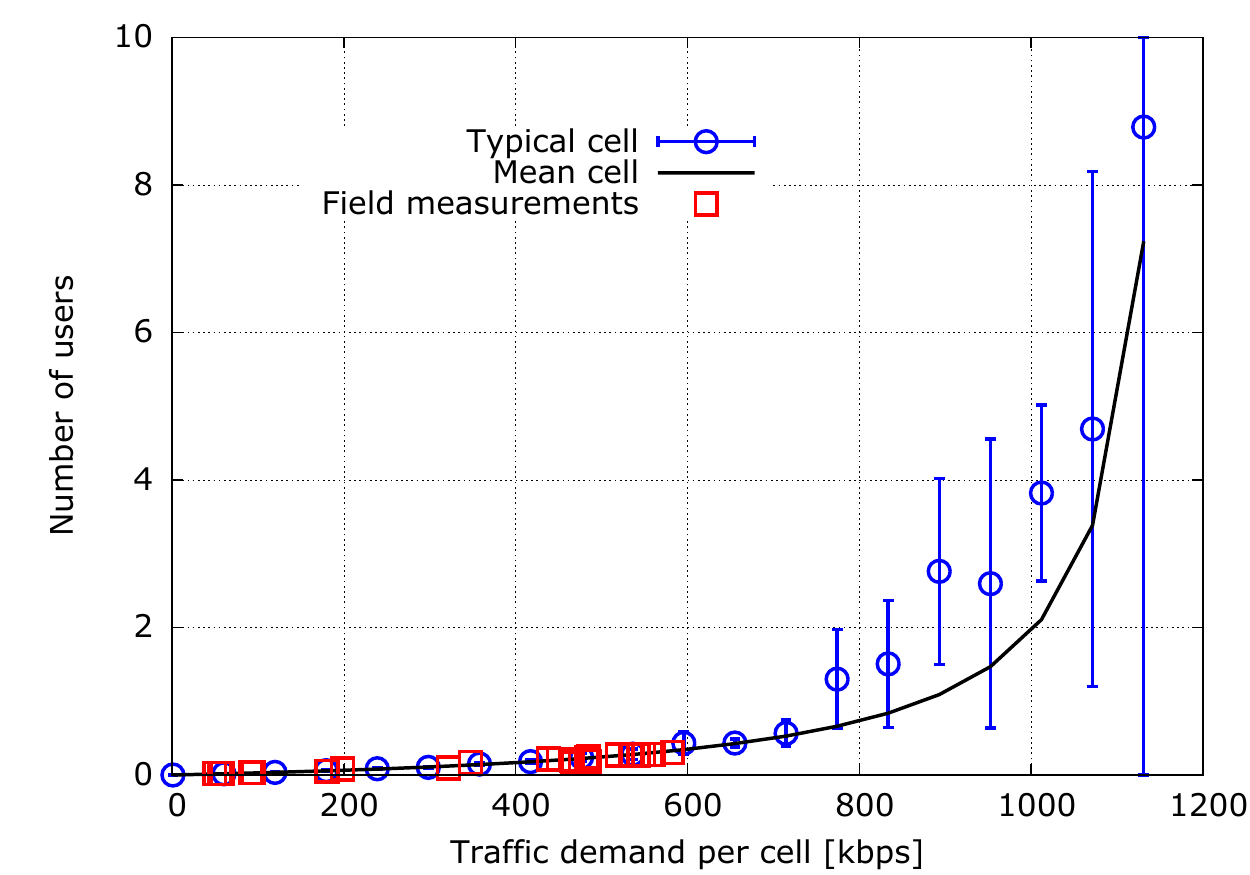}%
\caption{Number of users versus traffic demand per cell in the
  weighted interference model. Also, the same characteristic
estimated from the real field measurements.}%
\label{f.PonderedInterference_UsersNumber}%
\end{center}
%
%
\begin{center}
\includegraphics[width=0.8\linewidth]
{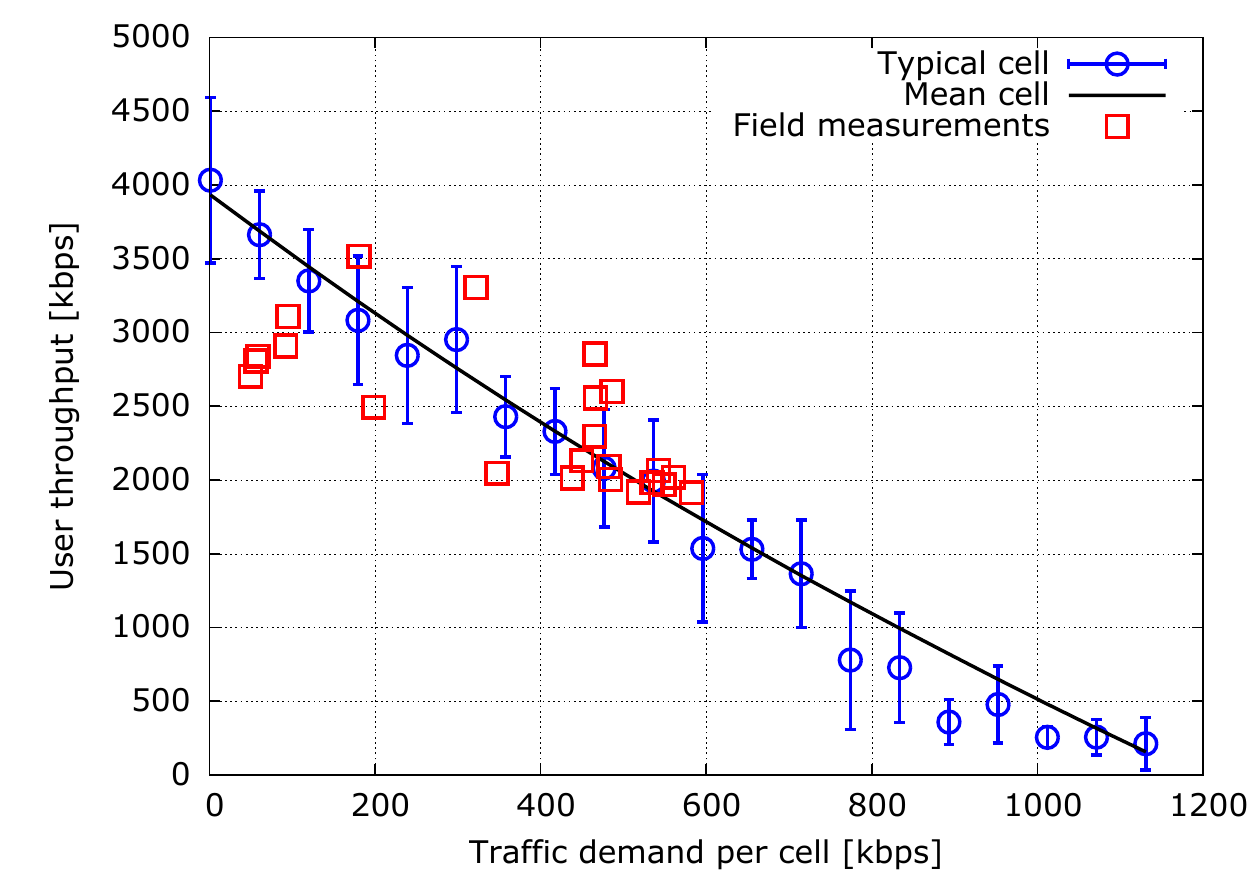}%
\caption{Mean user throughput in the network versus traffic demand per
  cell in the  weighted interference model. Also, the same characteristic
estimated from the real field measurements.}%
\label{f.PonderedInterference_UserThroughput}%
\vspace{-4ex}
\end{center}
\end{figure}

We consider first the full interference model (i.e. all BS emit the signal all
the time, regardless of whether or not they serve users).
Figure~\ref{f.FullInterference_Load} shows the mean cell load of the typical
cell $\mathbf{E}^{0}[\theta(0)]$ and the stable fraction of the network
$\pi_{\mathcal{S}}$ obtained from simulations, as well as the analytically
calculated load of the mean cell $\bar\theta$, versus mean traffic demand per
cell $\rho/\lambda_{\mathrm{BS}}$. We confirm that the typical cell and the
mean cell models have the same load. Note that for the traffic demand up to
500kbps per cell we do not observe unstable cells in our simulation window
($\pi_{\mathcal{S}}=1$).

Figure~\ref{f.FullInterference_UsersNumber} shows the mean number of users per
cell in the stable part of the network $N^{0}/\pi_{\mathcal{S}}$ (obtained
from simulations) and the analytically calculated number of users in the mean
cell $\bar N$ versus mean traffic demand per cell. We have two remarks. For
the traffic demand smaller than 500kbps per cell (for which all the simulated
cells are stable; $\pi_{\mathcal{S}}=1$, cf.
Figure~\ref{f.FullInterference_Load}), both models predict the same mean
number of users per cell. Beyond this value of the traffic demand per cell the
estimators of the number of users in the typical cell become not accurate due
very rapidly increasing fraction of the unstable region. (Error bars on all
figures represent the standard deviation in the averaging over 30 realizations
of the Poisson network).

Finally, Figure~\ref{f.FullInterference_UserThroughput} presents the
dependence of the mean user throughput in the network on the mean traffic
demand per cell obtained using the two approaches: $r^{0}$ and for the typical
cell and $\bar r$ for the mean cell. Again, both models predict the same
performance up to roughly 500kbps.

\subsection{Weighted interference}

We consider now the load-weighted interference model taking into account
idling cells. We see in Figures~\ref{f.PonderedInterference_Load},
\ref{f.PonderedInterference_UsersNumber} and
\ref{f.PonderedInterference_UserThroughput} that the consequence of this (more
realistic) assumption is that the cell loads are smaller, a larger fraction
$\pi_{\mathcal{S}}$ of the network remains stable, and the two approaches (by
the typical cell and by the mean cell) predict similar values of the QoS
metrics up to a larger value of the traffic demand per cell, roughly 700kbps.
Note that it is in this region that the real network operates for which we
present the measurements, and that its performance coincides with the
performance metrics calculated using the typical and mean cell approach.
More precisely, the field measurements on
Figures~\ref{f.PonderedInterference_Load},
\ref{f.PonderedInterference_UsersNumber} and
\ref{f.PonderedInterference_UserThroughput}
correspond to the same day and network zone considered on
Figure~\ref{f.crude-data}.

\begin{remark}
[Measurement methodology]\label{r.measure-method}
Measurement points
on Figure~\ref{f.PonderedInterference_Load}
show the fraction of time, within a given hour,
when the considered base stations were idle, averaged over the
base stations, as function of the average
traffic demand during this hour.
Similarly, measurement points on
Figure~\ref{f.PonderedInterference_UsersNumber}
show the spatial average of the mean number of users reported by
the considered base stations within a given hour, as function of the
average traffic demand during this hour.
Finally, measurements on Figure~\ref{f.PonderedInterference_UserThroughput}
give the ratio of the total number of bits transmitted by all the
base stations during a given hour, to the total number of users they  served
during this hour in function  of the
average traffic demand during this hour.
\end{remark}

Remark that
Figure~\ref{f.PonderedInterference_UserThroughput}  makes evident
a macroscopic relation between
the traffic demand and the mean user throughput in the
network zone already considered on  Figure~\ref{f.crude-data}.
This relation, we are  primarily looking for in this paper,
is not visible without the spatial averaging of the
network measurements described in Remark~\ref{r.measure-method}.
In order to ensure the reader that a
relatively good matching between the measurements and the analytic
prediction is not a coincidence we present on
Figure~\ref{f.SecondCity} similar results for a {\em urban zone} of a different European
city, where the spatial homogeneous Poissonianity of the base station
locations can also be retained; cf. Figure~\ref{f.Ripley}.
The only engineering difference of this network zone with respect
to the previously considered  dense urban zone is roughly four times smaller density of base stations,
more precisely 1.15 stations per km${}^2$.

\begin{remark}[Day and night hours]
Let us make a final remark regarding the empirical relation between the mean user throughput
on the mean traffic demand revealed on Figures~\ref{f.PonderedInterference_UserThroughput} and \ref{f.SecondCity}.
Recall that different points on these plots correspond to different hours of some given day. In fact, the points
laying below the mean curve correspond to day hours  while the
points laying above the mean curve correspond to night hours.
This ``circulation'' of the measured values
around the theoretical mean curve, indicated on Figure~\ref{f.SecondCity} and  visible on both presented plots of the throughput,
seems to be a more general rule,
which escapes from the analysis presented in this paper and remains an open question.
A possible explanation can lay in a  different space-time structure of
the traffic during the day and night, with the former one being much more clustered (fewer users, requesting larger volumes,
generating less interference and overhead traffic).
\end{remark}

\begin{figure}[t]
\begin{center}
\includegraphics[width=0.8\linewidth]{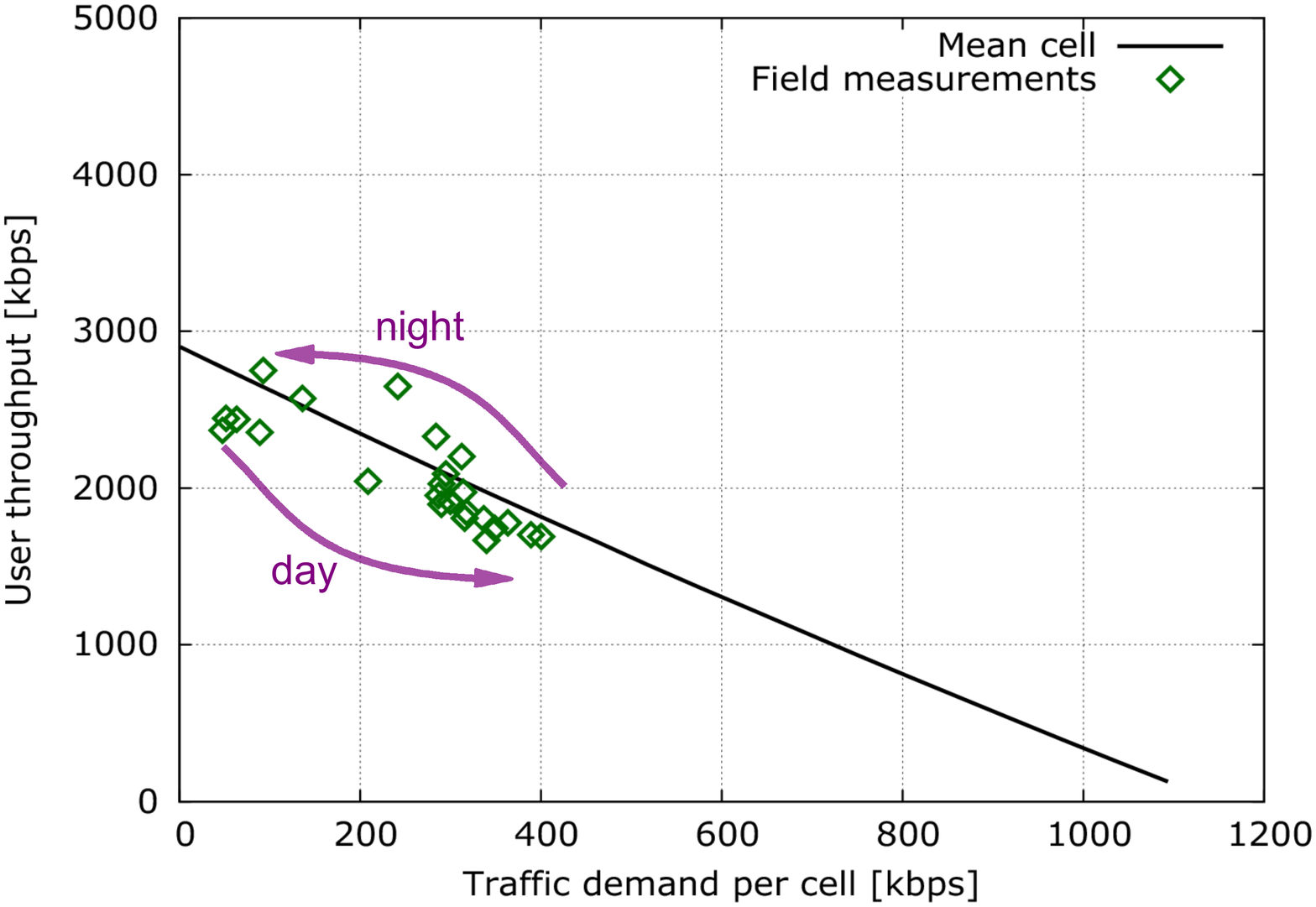}
\end{center}
\caption{Mean user throughput in the network versus traffic demand per surface
for an urban zone of a big city in Europe. (The density of base stations is 4
times smaller than in the dense urban zone considered on
Figure~{\ref{f.PonderedInterference_UserThroughput}}).}%
\vspace{-2ex}
\label{f.SecondCity}
\end{figure}

\begin{figure}[t]
\begin{center}
\includegraphics[width=1\linewidth]{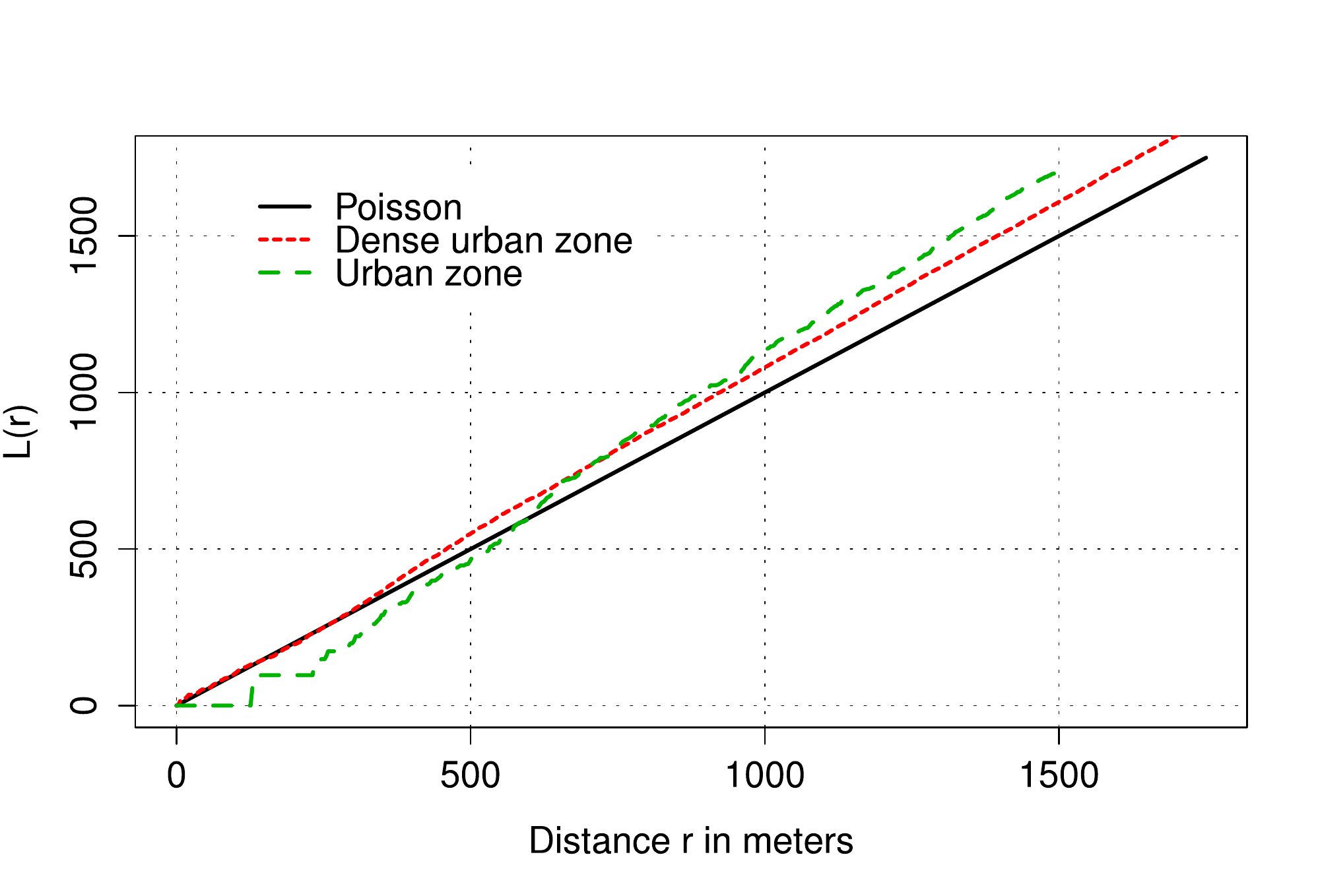}
\end{center}
\vspace{-2ex}
\caption{Ripley's $L$-function calculated for the considered
dense urban and urban network zones.
($L$ function is the square root of the sample-based estimator of the expected number of neigbours of the typical point
within a given distance, normalized by mean number of points in the
disk of the same radius. Slinvyak's theorem allows to calculate the
theoretical value of this function for a homogeneous
Poisson process, which is  $L(r)=r$.) In fact, in large  cities spatial, homogeneous ``Poissonianity'' of base-station   locations is often  satisfied
``per zone'' (city center, residential
zone, suburbs, etc.). Moreover, log-normal shadowing further justifies Poisson assumption, cf.~\protect\cite{brown2000cellular,BlaszczyszynKarrayKeeler2013}.}%
\label{f.Ripley}%
\vspace{-4ex}
\end{figure}

\section{Conclusion}

In order to evaluate user's QoS metrics, in particular the mean user
throughput, in large irregular multi-cellular networks, two approaches based
on stochastic geometry in conjunction with queueing and information theory are
developed. The typical cell approach consists in considering true spatial
averages of local network characteristics and thus capturing the global
network performance. A simpler, approximate but fully analytic approach,
called the mean cell approach, is inspired by the analysis of the typical
cell. The key quantity explicitly calculated in this approach is the cell
load. In analogy to the load factor of the (classical) M/G/1 processor sharing
queue, it characterizes the stability condition, mean number of users and the
mean user throughput. Considering some real field measurement, we validate the
proposed approach by showing that it allows to predict the performance of a
real network.

The present work raised open theoretical questions regarding the stability of
spatially and, more difficult, space-time dependent processor sharing queues
modeling the performance of individual network cells (cf
Section~\ref{s.PonderedInterference}). More work is also required to figure
out the problem of different performance of the network during day and night
hours (cf Figure~\ref{f.SecondCity}).

\bibliographystyle{IEEEtran}

\end{document}